\documentclass[conference]{IEEEtran}
\IEEEoverridecommandlockouts
\usepackage[top=0.752in, bottom=1.05in, left=0.680in, right=0.673in]{geometry}

\usepackage{cite}
\usepackage{amsmath,amssymb,amsfonts,mathrsfs,amsthm}
\usepackage{algorithm}
\usepackage{algpseudocode}
\usepackage{multirow}
\usepackage{tikz}
\usepackage{graphicx}
\usepackage{textcomp}
\usepackage{xcolor}
\usepackage{url}
\usepackage{adjustbox}
\usepackage{balance}
\usepackage{makecell}
\newtheorem{theorem}{Theorem}

\DeclareMathOperator*{\argmin}{arg\,min}
\begin{document}

\title{SINR-Aware  Base Station Deployment in\\ Wide Area IoT Sensor Networks}

\author{\IEEEauthorblockN{Sachin Kadam}
\IEEEauthorblockA{{Department of Electronics and Communication Engineering} \\
{Motilal Nehru National Institute of Technology (MNNIT) Allahabad, Prayagraj, UP 211004, India}\\
Email: sachink@mnnit.ac.in}
}

\maketitle

\begin{abstract}
The rapid expansion of Internet of Things (IoT) applications necessitates the effective deployment of base stations (BSs) to enable consistent connectivity across large geographic areas under interference-limited conditions. Existing techniques typically use distance-based or binary coverage models; however, these abstractions fail to account for the influence of co-channel interference on the quality of communication in dense deployments. In this paper, we investigate the Signal-to-Interference-plus-Noise Ratio (SINR)-aware Base Station Deployment (BSD) problem in wide-area IoT sensor networks. The objective is to determine a minimum-cost subset of BSs from a predefined set of candidate BSs such that every IoT sensor is covered by at least one BS and a target SINR threshold is satisfied. The problem is formulated as a combinatorial optimization problem, which is NP-hard. Theoretical analysis establishes that the proposed coverage function is monotone and submodular, enabling the SINR-aware greedy algorithm to achieve a ($1\text{-}1/e$)-approximation to the optimal solution while maintaining a polynomial-time computational complexity. Numerical evaluations on a real water distribution network dataset demonstrate that the proposed SINR-aware greedy algorithm achieves near-optimal base station deployment while significantly reducing computational effort. Compared with the Genetic Algorithm (GA) and Particle Swarm Optimization (PSO) algorithms, the proposed approach attains complete sensor coverage with deployment costs within 12.3\% of the best-performing metaheuristic solution while requiring up to 190 times lower execution time.
\end{abstract}


\section{Introduction}
The growth of Internet of Things (IoT) technologies and wireless sensor networks (WSNs) has enabled large-scale monitoring and automation in a wide range of application domains, including precision agriculture, environmental monitoring, and industrial systems. These technologies enable diverse applications such as intelligent irrigation, soil moisture monitoring, optimized fertilizer usage, early detection of pests and crop diseases, and efficient energy management~\cite{mowla2023internet}. In such wide-area IoT deployments, a massive number of spatially distributed sensors generate data that must be reliably collected by infrastructure nodes, typically base stations (BSs). As network scale increases, ensuring robust connectivity and quality-of-service (QoS) becomes a critical challenge, particularly in dense deployments with heterogeneous traffic demands~\cite{nb_iot_2024,ris_bs_2023}. In such wide-area IoT networks, the strategic deployment of base stations is important for ensuring connectivity and coverage for geographically dispersed sensors. Furthermore, IoT sensors (IoT devices interfaced with sensors) are often energy-constrained, requiring careful management of transmission power to maintain network longevity~\cite{cai2024joint}. As a result, base station placement selections must consider coverage, energy efficiency, and system cost. In practice, base stations are deployed at specific candidate sites such that their combined coverage includes all IoT sensors.
\begin{figure}
    \centering  \includegraphics[width=0.9\linewidth]{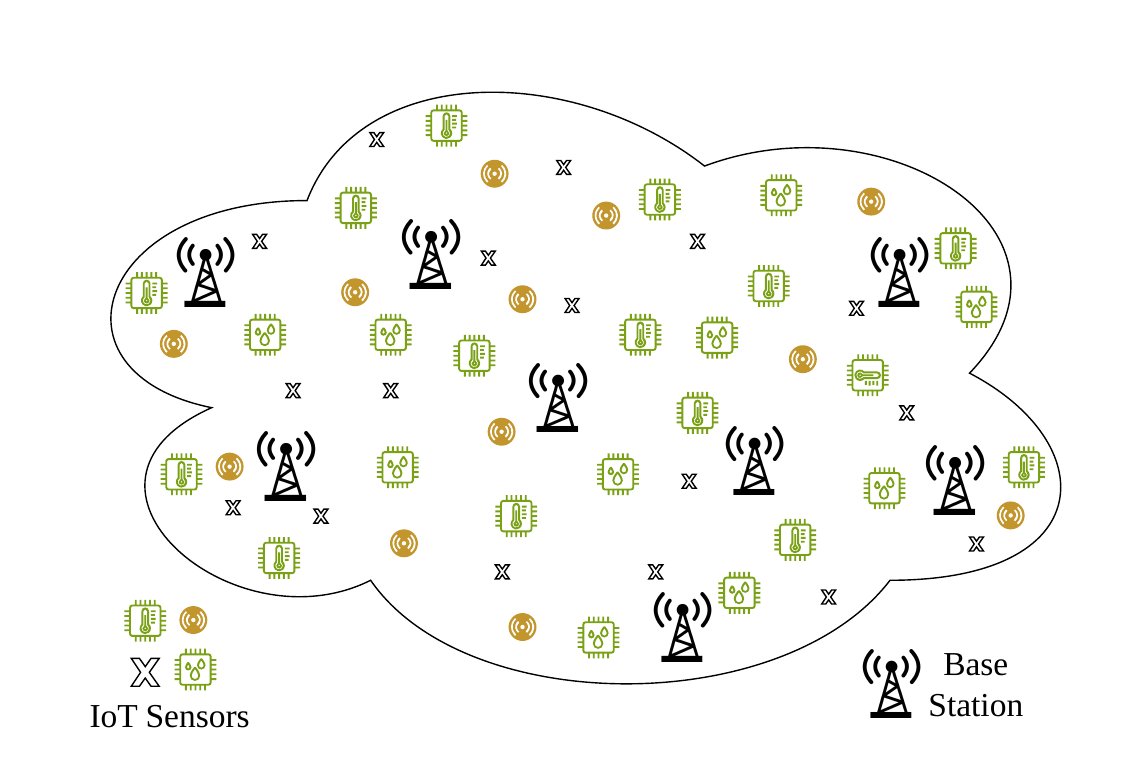}
    \caption{Illustration of a base station deployment (BSD) in a wide area region such that the combined coverage of these BSs includes all deployed IoT sensors.}
    \label{fig:MultiBS}
\end{figure}

Co-channel interference is a critical element impacting communication reliability in dense IoT sensor networks, as it has a direct impact on the receiver's signal-to-interference-plus-noise ratio (SINR). SINR is generally regarded as a crucial metric for evaluating link dependability and communication performance in current wireless communication systems. It captures the combined impacts of interference, noise, and channel conditions~\cite{sinr_model_2024}. Unlike traditional coverage models based solely on distance or binary coverage models, SINR-based models provide a more realistic representation of wireless communication quality, especially in interference-limited environments.

In this paper, we investigate the SINR-aware Base Station Deployment (BSD) problem in wide-area IoT sensor networks.
We formulate an optimization problem whose objective is to find the optimal set of BSs from a set of candidate BSs to cover the complete set of IoT sensors in the field with the minimum cost, measured in terms of the number of BSs. Note that the proposed cost model serves as a proxy for deployment efficiency, where BSs covering a larger number of sensors incur lower effective costs. Consequently, minimizing the aggregate cost also promotes deployments with fewer, strategically positioned BSs. Solving this class of problems exactly requires a combinatorial search. However, recognizing that this problem formulation falls in the so-called class of submodular optimization problems, we propose the use of a greedy algorithm to solve it since, for submodular problems, there are guarantees that the solution will provide a result for the objective that differs from the optimum one by no more than $(1-1/e)$~\cite{schrijver2003combinatorial}. This distinguishes the proposed formulation and the solution from conventional coverage-based deployment models and enables more realistic and reliable network design in interference-limited IoT environments. 

The remainder of this paper is organized as follows: Section II reviews state-of-the-art literature work on base station deployment and IoT network optimization. Section III presents the system model and problem formulation. Section IV and Section V discuss the analytical results and numerical results, respectively. Finally, Section VI concludes the paper and outlines a few future research directions.

\section{Related Work}
Recent advances in WSNs and IoT applications have highlighted the significance of base station deployment in providing dependable communication, adequate coverage, long-term network operation, and scalable system architecture. As a result, extensive research has been conducted to develop effective deployment and mobility strategies for BSs in large-scale sensing settings.

Initial research centered mostly on increasing network longevity through energy-efficient deployment technologies. In ~\cite{hu2020co}, the authors explored a hybrid architecture that included a terrestrial BS and a UAV-assisted aerial BS. Their approach enhanced coverage partitioning and bandwidth allocation to reduce sensor nodes' uplink transmission energy, resulting in significant energy savings over traditional deployment strategies.

Apart from deployment optimization, researchers investigated the impact of BS architecture and mobility on network performance. A multi-sector antenna system was developed in ~\cite{oliveira2022implementation} for environmental monitoring applications, ensuring reliable communication at distances greater than 3.5 km and providing comprehensive coverage. On a more general level, the combination of WSNs and IoT systems raises new issues on localization, routing and security. These elements have been identified as the key determinants of network efficiency and reliability in the survey presented in~\cite{vishwas2025recent}. It underscores the necessity of precise localization methods, efficient routing protocols, and strong security mechanisms for extensive IoT implementation.

Recently, researches on base station deployment have been focused on optimization frameworks that aim at multiple performance objectives simultaneously. In~\cite{zhu2023edge_deployment} the authors studied the deployment solutions for edge-enabled sensor networks with coverage and reliability constraints. Similarly, the work in~\cite{li2025bsd_dynamic} highlighted the importance of realistic spatial modeling for deployment planning in dynamic environments. To cope with the intrinsic complexity of deployment problems, many optimization based approaches have been proposed. In~\cite{chen2023multiobjective}, authors introduced a multi-objective framework to balance deployment cost and network reliability, while in~\cite{zaimen2024hybrid}, a hybrid metaheuristic approach was utilized to optimize deployment in complex scenarios. Recent survey studies~\cite{tossa2025survey} reinforce the direct impact of deployment decisions on the network scalability, energy efficiency, and long-term operational efficacy in IoT systems.

Recent studies have looked at how to optimize deployments in IoT and WSNs in terms of coverage, cost, and reliability. For example, next-generation IoT systems have shown significant gains in coverage and capacity with RIS-assisted networks and dense multi-base station architectures~\cite{ris_bs_2023,bie2024ris}. But many existing methods are based on simplified communication models which ignore interference or assume binary coverage conditions. These assumptions can result in sub-optimal or infeasible deployments in wide-area dense networks due to severe interference and link quality degradation.

Mobile base stations (MBSs) are proposed as an alternative to the traditional fixed deployment for wide-area IoT applications. The studies presented in~\cite{kadam2020spcom,KADAM2024103779,kadam2026mobilebasestationoptimal} have shown that MBS-assisted designs can potentially improve data collection efficiency and extend network coverage in geographically distributed sensing environments. These results indicate that the integration of mobility into network architecture is a promising approach to solve the challenges involved with large-scale deployment of IoT sensors.

Despite these major gains, current research has mostly focused on energy efficiency, coverage maximization, deployment cost reduction, and network lifetime enhancement. Furthermore, while alternative approaches enhance deployment efficiency, most recent works use simplified communication models that assume binary coverage or ignore interference effects. During the base station deployment process, there has been relatively little emphasis on directly incorporating signal quality measures, particularly the SINR. It has a direct impact on communication reliability, data delivery performance, and network robustness in dense and large-scale IoT environments, so deployment strategies that explicitly account for interference-aware communication conditions remain an important and largely unexplored research direction. This gap motivates the development of SINR-aware base station deployment frameworks for wide-area IoT sensor networks. Furthermore, this integrated formulation more accurately reflects practical deployment scenarios.

\section{System Model and Problem Formulation}
Let $\mathcal{B}$ be the set of given BSs, with $\beta_b$ being the location of BS $b$. Similarly, $\mathcal{S}$ be the set of all deployed IoT sensors with $\xi_s$ being the location of IoT sensor $s$ in a given region. So, $\mathcal{B} = \{\beta_b = (\bar{x}_b,\bar{y}_b) : b = 1, \ldots, B\}$ and $\mathcal{S} = \{\xi_s = (x_s,y_s) : s = 1, \ldots, S\}$, where $B = |\mathcal{B}|$ and $S = |\mathcal{S}|$. 
As shown in Fig.~\ref{fig:MultiBS}, the distance between a given BS location $\beta_b$ and a given sensor location $\xi_s$ is $d_{bs} = \sqrt{(x_s - \bar{x}_b)^2 + (y_s - \bar{y}_b)^2 + (h_b - h_s)^2}$, where $h_b$ and $h_s$ are the heights of BS $b$ and sensor $s$, respectively.
\begin{figure}
    \centering  \includegraphics[width=0.86\linewidth]{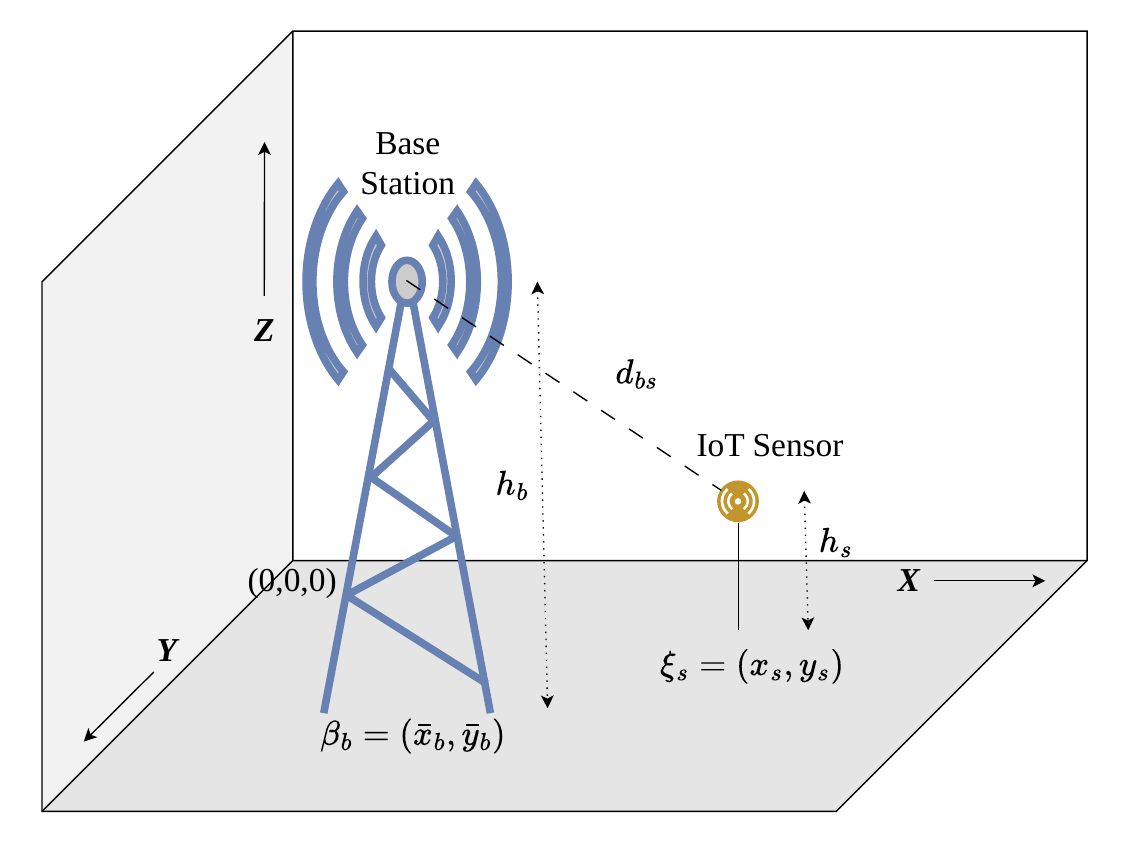}
    \caption{The distance $d_{bs}$ between a base station $\beta_b$ and an IoT sensor $\xi_s$ accounts for both horizontal separation and height difference.}
    \label{fig:bsd_model}
\end{figure}
The received power at the sensor $s$ from BS $b$ is:
\begin{equation}
P_{r}(b,s) = P_t (b) d_{bs}^{-\alpha} r_{bs}
\end{equation}
where $P_t (b)$ is the transmission power of BS $b$, $\alpha$ is the path-loss exponent and $r_{bs}$ models Rayleigh fading. Next, we compute the SINR at the sensor $s$ from the base station $b$   using the following equation:
\begin{equation}
\text{SINR}_{bs} = \frac{P_r(b,s)}{\sum_{b' \neq b} P_r(b',s) + \sigma^2}
\end{equation}
where $\sigma^2$ corresponds to the AWGN power. A sensor $s$ is considered covered by BS $b$ if:
\begin{equation}
\text{SINR}_{bs} \geq \gamma_{th},
\end{equation}
where $\gamma_{th}$ denotes a predetermined threshold value.

Let $X_{bs}\in~\{0,1\}$ be equal to 1 if the sensor located at $\xi_s$ is in the coverage range of the BS deployed at $\beta_b$, and zero otherwise. Mathematically, $X_{bs}$, $\forall s \in \mathcal{S}, \forall b \in \mathcal{B}$ can be defined as:
\begin{align}
    X_{bs} := 
    \begin{cases}
    1, &\mbox{if } \text{SINR}_{bs} \geq \gamma_{th} \\ 
    0, &\mbox{otherwise}.
    \end{cases}
\end{align}
Let $\mathcal{X}_b$ be the set of sensors that are in the coverage range of BS $b$. Mathematically, $\mathcal{X}_b$, $\forall b \in \mathcal{B}$ can be defined as:
\begin{align}
    \mathcal{X}_b := \{ s | X_{bs} =1, \forall s \in \mathcal{S} \}.
\end{align}
For any subset $\mathcal{A} \subseteq \mathcal{B}$, we define the coverage function  $f(\mathcal{A})$ as:
\begin{equation}
    f(\mathcal{A}) := \left|\bigcup_{a \in \mathcal{A}} \mathcal{X}_a\right|.
    \label{eq:coverage}
\end{equation}
It denotes the number of sensors in the coverage range of BSs in $\mathcal{A}$. 
Let $\psi_b = |\mathcal{X}_b|$. So $\psi_b = \sum_{s \in \mathcal{S}} X_{bs}$ denotes the number of sensors in the coverage range of BS $b$.  
Let $c_b$ be the cost of installing the BS $b \in \mathcal{B}$ which approximately decreases as the expected number of sensors in its range increases. 
That is, $c_b = \psi_b^{-\eta},$ for $\eta > 0$. From this definition, we see that $f(\mathcal{B}) = S$. Also, 
\begin{equation}
    \psi_b (\mathcal{A}) = f(\mathcal{A} \cup b) - f(\mathcal{A}).
    \label{eq:psi_b}
\end{equation}

We are interested in finding the set with the least number of BSs, set $\mathcal{A}$, called BSD,  from a given set of candidate BSs and for a given sensor deployment such that each sensor is in the coverage range of at least one BS. Now, we define the optimal base station deployment (BSD) problem as follows:
\begin{subequations}
\begin{align} \label{eq:Max_Cover}
\argmin_{\substack{{\mathcal{A}}}} & \sum_{a \in \mathcal{A}} c_{a}   \\
\text{ s.t. } & f(\mathcal{A}) = S, ~S \in \mathcal{N},\label{eq:coverageconstraint}\\ 
&\mathcal{A} \subseteq \mathcal{B} . \label{eq:subsetconstraint}
\end{align}
\end{subequations}
In this BSD problem, we are finding the minimal cost BS subset $\mathcal{A}$ with the size $A = |\mathcal{A}|$ subject to full coverage of all the sensors by set $\mathcal{A}$ (see~\eqref{eq:coverageconstraint}) and it should be a subset of the given set $\mathcal{B}$ (see~\eqref{eq:subsetconstraint}).

\section{Analytical Results}
\begin{theorem}\label{thm:sub}
    $f(\mathcal{A})$ is a monotone and submodular function.
\end{theorem}
\begin{proof}
The proof is given in Appendix~\ref{Apdx:thmsub}.
\end{proof}

Now, we propose a SINR-aware greedy algorithm (see Algorithm~\ref{alg:greedy}) to solve the optimization problem with submodular objective function shown in~\eqref{eq:Max_Cover}. 
\begin{algorithm}
\caption{SINR-aware Greedy Algorithm}\label{alg:greedy}
\begin{algorithmic}
\State $\mathcal{G}^0 \gets \emptyset$, $\mathcal{S}^0 \gets \mathcal{S}$
\State $i \gets 1$
\While{($i \neq K$)  AND  ($\mathcal{S}^{i-1} \neq \emptyset$)}
\While{${a'\in \mathcal{S}^{i-1}}$} \Comment{SINR-aware step}
\State $\psi_{a'}(\mathcal{G}^{i-1}) = f(\mathcal{G}^{i-1} \cup \{a'\}) - f(\mathcal{G}^{i-1})$ 
\EndWhile
\State $\{a_i\} \gets \arg \min_{\{a'\}} c_{a'}(\mathcal{G}^{i-1})$ \Comment{Greedy selection}
\State $\mathcal{G}^{i} \gets \mathcal{G}^{i-1} \cup \{a_i\}$
\State $\mathcal{S}^{i} \gets \mathcal{S}^{i-1} \setminus \{a_i\}$
\State $i \gets i+1$
\EndWhile
\end{algorithmic}
\end{algorithm}

Let $\mathcal{G}$, with cardinality $G = |\mathcal{G}|$, be the set of BSs obtained using the greedy algorithm (Algorithm~\ref{alg:greedy}) and it is  initialized to a null set. Now, from~\eqref{eq:psi_b}, we see that  $\psi_{a'}(\mathcal{G}) = f(\mathcal{G} \cup \{a'\}) - f(\mathcal{G})$, for any singleton set $\{a'\} \in \mathcal{S}$. We check for each $\psi_{a'}$ and search for the subset $a_i$ that minimizes the cost $c_{a'}$.\footnote{Note that $c_a = \psi_a^{-\eta}$ for $\eta > 0$. This relationship represents an effective deployment utility metric rather than literal monetary cost.} 
Now, we solve our optimal BSD problem by recursively executing Algorithm~\ref{alg:greedy} for $K = 1, 2, \ldots$, until complete coverage is achieved. That is $\mathcal{G}^{G} = \mathcal{S}$ and also $\mathcal{S}^{G} = \emptyset$. 

For the sake of comparison with the proposed greedy algorithm in the numerical results that follow,  we compute an optimum solution $\mathcal{A^{\star}}$, with cardinality $A^{\star} = |\mathcal{A^{\star}}|$, through exhaustive search, which is still manageable for small values of $S$ and $B$. To reduce the number of computations, we perform the search operation as follows:
First we start with $A=1$, i.e., consider one station at a time and check whether it can cover the entire region. If yes, then the optimum solution size $A^{\star} = 1$, else we increase $A=2$. Now consider all possible combinations of 2 BSs taken at a time and check whether they can cover the entire region. If yes, then the optimum solution size  $A^{\star} = 2$, else we increase $A=3$. Similar procedure is continued until at $A = \bar{B}$, at least one set of $\bar{B}$ BSs covers the entire region, which is the solution size $A^{\star} = \bar{B}$. 

\begin{theorem}\label{thm:greedy}
The proposed SINR-aware greedy algorithm achieves a $(1-1/e)$ approximation guarantee with respect to the optimal solution.
\end{theorem}
\begin{proof}
The proof is given in Appendix~\ref{Apdx:greedy}.
\end{proof}
\begin{theorem}\label{thm:complexity} The proposed SINR-aware greedy algorithm has time complexity:
\begin{equation}
\mathcal{O}(K B S)
\end{equation}
and space complexity:
\begin{equation}
\mathcal{O}(B + S)
\end{equation}
\end{theorem}
\begin{proof}
The proof is given in Appendix~\ref{Apdx:complexity}.
\end{proof}

Compared to the exhaustive search, whose time complexity is $\mathcal{O}(2^S)$, the proposed greedy algorithm runs much faster.
The greedy approach provides a scalable and efficient solution with near-optimal performance. It is suitable for large-scale IoT deployments.

\section{Numerical Results}\label{Sec:Results}
\begin{figure}
    \centering 
    \includegraphics[width=0.95\linewidth]{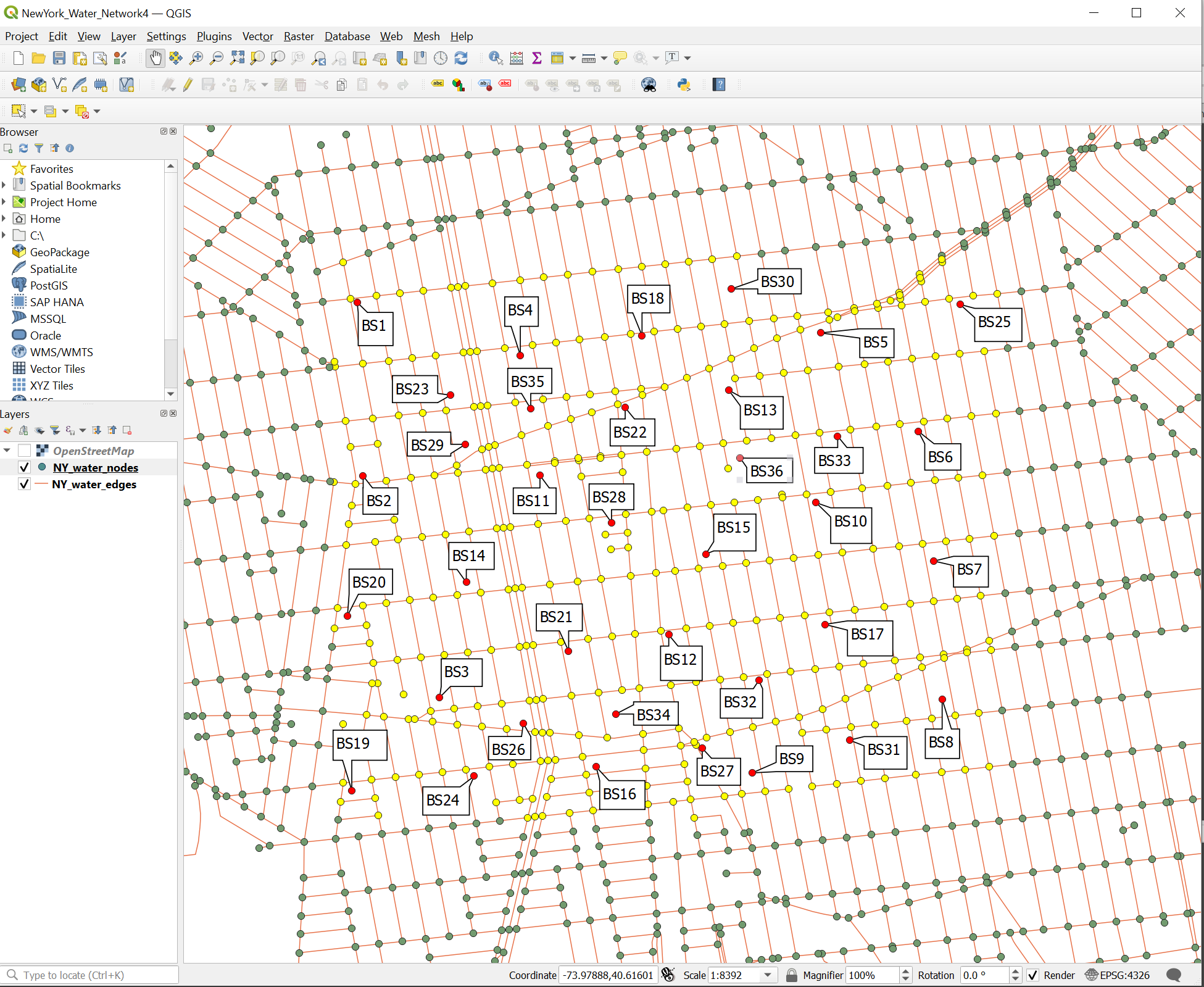}
    \caption{We have considered $ 2 \times 2$ km$^2$ of the WDN wherein small yellow circles denote the IoT sensors (which are used in our evaluations), small green circles also denote the IoT sensors (which are not used in our evaluations), red lines denote the water pipes connecting the IoT sensors, and small red circles denote the candidate BS deployment locations.  
    }
    \label{fig:NewYork_city}
\end{figure}
In this section, we provide an example of optimal BS deployment in a water distribution network (WDN). 
In this context, we provide the comparison of the solutions obtained from both the proposed SINR-aware greedy algorithm and the exhaustive search methods, considering the synthetic data for a WDN of New York city~\cite{PipeNetwork} in Fig.~\ref{fig:NewYork_city}.  For ease of implementation, we consider placing the IoT sensors at the places of water network nodes, which represent the junctions between the water pipes.  We consider the (2 km $\times$ 2 km) grid area as shown in Fig.~\ref{fig:NewYork_city} and the urban Non-Line-of-Sight (NLOS) path loss model~\cite{jiang20213gpp}. We use the `Haversine formula'~\cite{sinnott1984virtues} to compute the distance between sensors and BSs with GIS coordinates. The Haversine formula computes the great-circle distance between two points on a sphere using their latitudes and longitudes.

Let:
\begin{itemize}
    \item $\phi_1, \phi_2$ be the latitudes of points 1 and 2 (in radians),
    \item $\lambda_1, \lambda_2$ be the longitudes of points 1 and 2 (in radians),
    \item $R$ be the radius of the sphere (for Earth, $R \approx 6371\,\text{km}$).
\end{itemize}
\begin{subequations}
\begin{align}
\Delta \phi &= \phi_2-\phi_1,\\
\Delta \lambda &= \lambda_2-\lambda_1,\\
a &= \sin^2\!\left(\frac{\Delta\phi}{2}\right)
+\cos(\phi_1)\cos(\phi_2)
\sin^2\!\left(\frac{\Delta\lambda}{2}\right),\\
c &= 2\arctan2(\sqrt{a},\sqrt{1-a}),\\
d &= Rc.
\end{align}
\end{subequations}
Note that all angular quantities must be expressed in radians. The formula assumes a spherical Earth.\footnote{For higher precision geodetic calculations on an ellipsoid, more advanced models (e.g., Vincenty's formulae) are used.} 
The parameters used are: $S = 358$, $B=36$, $P_t=20$ dBm, $\sigma=4$, $r_{bs} = 1$,\footnote{For simplicity and reproducibility, the numerical results consider deterministic channel gains ($r_{bs} = 1$); incorporating stochastic fading effects constitutes an important direction for future work.} $\gamma_{th}=-50$ dB, and $\alpha=2.8$  for the given IoT water distribution network~\cite{jiang20213gpp}. 

\begin{figure}
    \centering
    \includegraphics[width=0.99\linewidth]{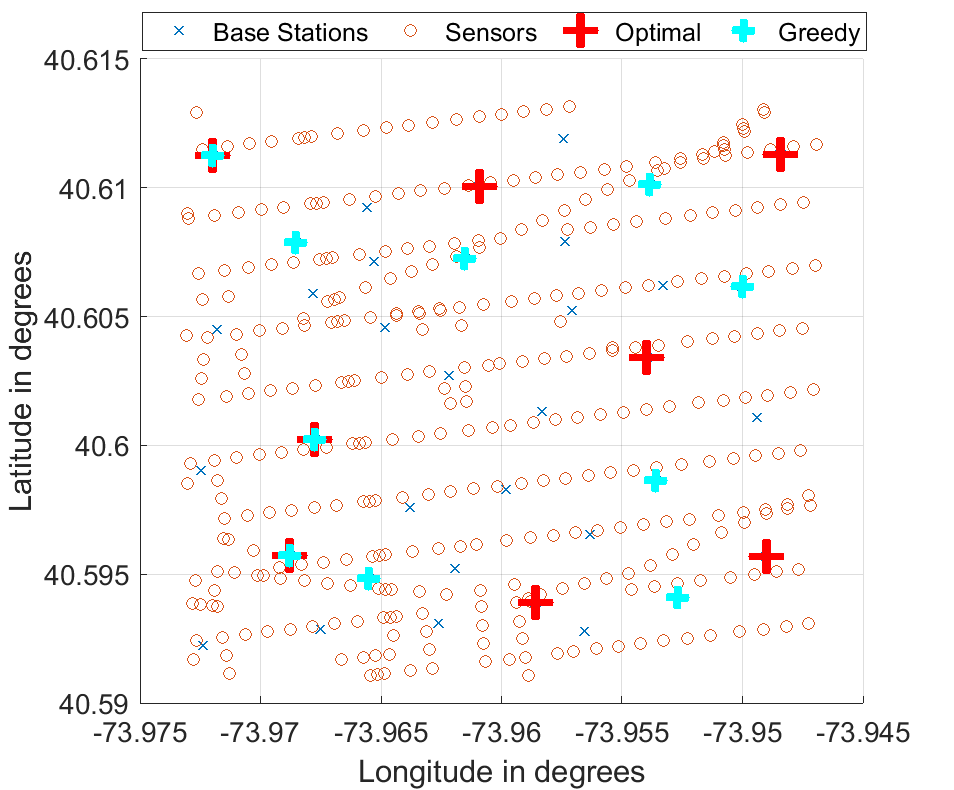}
    \caption{The sensors are marked as red $\circ$, the candidate BSs are marked as blue $\times$ while the red $+$ and turquoise $+$ are the optimal and greedy solutions, respectively. The optimal solution requires 8 BSs, whereas the proposed SINR-aware greedy algorithm solution requires 10, and only three of the options are the same between the two.}    \label{fig:GreedyOptCompare}
\end{figure}
We now apply both the proposed SINR-aware greedy algorithm and exhaustive search methods to this setup to find the optimal number of BSs required. The results are shown in Fig.~\ref{fig:GreedyOptCompare}. In this scatter plot,  sensors and BS locations with respect to GIS coordinates are  shown. The optimal solution requires 8 BSs  whereas the greedy algorithm solution requires 10 of them, and only three of the choices are identical between the two.
Note that for ease of visualization, we have shown only a part of New York city water distribution network. The similar procedure can be repeated to cover the entire city.

\begin{table}
\caption{Performance comparison of the proposed SINR-aware greedy algorithm with PSO and GA.}
\label{tab:comparison25}
\centering
\begin{tabular}{|c|c|c|c|}
\hline
Algorithm &
Selected BSs &
Total Cost &
Runtime (s) \\
\hline
SINR-Aware Greedy & 10 & 0.82 & 0.03 \\
PSO~\cite{kennedy1995particle} & 9 & 0.76 & 2.21 \\
GA~\cite{goldberg2013genetic} & 9 & 0.73 & 5.71 \\
\hline
\end{tabular}
\label{Tab:GA_PSO_Compare}
\end{table}
Next, we compare the proposed SINR-aware greedy algorithm with the Particle Swarm Optimization (PSO) algorithm and Genetic Algorithm (GA) in terms of total cost and runtime needed to execute the algorithm and the results are shown in Table~\ref{Tab:GA_PSO_Compare}. 
The results demonstrate that all three algorithms successfully achieve complete coverage of the 358 sensors. The greedy algorithm requires 10 BSs with a deployment cost of 0.82, whereas PSO and GA reduce the cost to 0.76 and 0.73, respectively, by selecting 9 BSs. Nevertheless, these cost savings come with a much larger computational overhead. Specifically, the execution time of PSO and GA is approximately 74 and 190 times greater than that of the greedy approach. Similar trends were observed across multiple random sensor deployments (results omitted due to space constraints). The omitted results consistently exhibited the same ordering among the compared algorithms.

Although GA provides the lowest deployment cost, the improvement over the proposed greedy approach is relatively small, corresponding to a reduction of approximately 11\% in deployment cost. In contrast, the greedy algorithm delivers solutions within 12.3\% of the GA cost while maintaining substantially lower execution time. Therefore, the proposed SINR-aware greedy algorithm offers an attractive tradeoff between deployment optimality and computational efficiency, making it particularly suitable for wide-area IoT sensor network planning where rapid decision-making and scalability are essential.

\section{Conclusions and Future Work}\label{Sec:Conclusions}
This paper addressed the challenge of SINR-aware Base Station Deployment (BSD) in wide-area IoT sensor networks by identifying a minimum-cost subset of base stations that provides complete sensor coverage while meeting a specified SINR criterion. Unlike traditional distance-based deployment methodologies, the proposed formulation explicitly included the impacts of interference, resulting in a more accurate representation of communication reliability in dense IoT contexts. The BSD problem was stated as a combinatorial optimization problem, and the resulting coverage function was found to have monotonicity and submodularity features. Leveraging these properties, a SINR-aware greedy algorithm was proposed with a $(1-1/e)$-approximation guarantee and polynomial-time computational complexity.

Numerical evaluations on a water distribution network in New York City demonstrated the effectiveness of the proposed methodology. The results show that the greedy approach provides deployment solutions relatively close to the optimal solutions obtained by the exhaustive search approach, at a significantly lower computational cost. In addition, the proposed approach obtained total sensor coverage with deployment costs within (12.3\%) of the best-performing metaheuristic solution when compared to GA and PSO methods while reducing the execution time by up to 190 times. Such results indicate that the proposed framework has a good trade-off between deployment quality and computational efficiency, making it a good candidate for large-scale IoT applications requiring scalability and quick decision-making.

Further avenues of research are still open. We can extend the current architecture to include heterogeneous base stations with varying transmission capacities and installation costs. Dynamic traffic patterns, energy-aware deployment considerations, and mobile base stations can enhance deployment flexibility in evolving IoT environments. Future work may consider resilient and stochastic optimization algorithms to handle uncertainties due to channel variations, sensor failures, and time-varying interference conditions. Moreover, an interesting approach to support next-generation smart city and critical infrastructure applications is to enhance the base station deployment along with resource allocation and network resilience mechanisms.

\appendices
\section{Proof of Theorem~\ref{thm:sub}}\label{Apdx:thmsub}
From~\eqref{eq:coverage}, the coverage function is written as:
\begin{equation}
f(\mathcal{A})
=
\left|
\bigcup_{a\in\mathcal{A}}
\mathcal{X}_a
\right|.
\end{equation}
Adding a new base station increases SINR coverage probability but with diminishing marginal gain due to interference coupling and saturation effects. Under independent fading and bounded interference assumptions, the marginal gain satisfies:
\begin{equation}
\Delta f(\mathcal{A}) \geq \Delta f(\mathcal{B}), \quad \mathcal{A} \subseteq \mathcal{B}
\end{equation}
Thus, $f(\mathcal{A})$ is monotone and submodular function.
\qed
\section{Proof of Theorem~\ref{thm:greedy}}\label{Apdx:greedy}
Suppose $\mathcal{A}^*$ is the optimal set and $\mathcal{A}_i$ is the greedy solution after $i$ iterations. 
Algorithm~\ref{alg:greedy} selects precisely $K$ BSs.
We demonstrate through mathematical induction that for $0 \le i \le K$,
\begin{equation}
f(\mathcal{A}^*)-f(\mathcal{A}_i)\le \left(1-\frac{1}{K}\right)^i f(\mathcal{A}^*).
\label{eq:induction}
\end{equation}

The base case of $i=0$ is trivially true. Let $i>0$. In the $i$-th step, Algorithm~\ref{alg:greedy} selects BS $a_i$, maximizing $f_{\mathcal{A}_{i-1}}(a_i)$ among the remaining set of BSs (see line 7 of Algorithm~\ref{alg:greedy}).\footnote{Since the proposed cost decreases monotonically with the number of covered sensors, maximizing the incremental coverage gain tends to favor lower-cost base station selections.} Note that the remaining BSs comprise $\mathcal{A}^*\setminus \mathcal{A}_{i-1}$, a set with a maximum size of $K$. Submodularity allows us to
\[
f(\mathcal{A}^*)-f(\mathcal{A}_{i-1})
\le
\sum_{a_i\in \mathcal{A}^*\setminus \mathcal{A}_{i-1}} f_{\mathcal{A}_{i-1}}(a_i),
\]
and this implies that the BS $a_i$ has marginal value
\begin{align}
f_{\mathcal{A}_{i-1}}(a_i)
&\ge
\frac{1}{|\mathcal{A}^*\setminus \mathcal{A}_{i-1}|}
\sum_{a_i\in \mathcal{A}^*\setminus \mathcal{A}_{i-1}} f_{\mathcal{A}_{i-1}}(a_i) \nonumber\\
&\ge
\frac{1}{K}\bigl(f(\mathcal{A}^*)-f(\mathcal{A}_{i-1})\bigr).
\end{align}

Assuming that~\eqref{eq:induction} holds true for $\mathcal{A}_{i-1}$, we have
\begin{align*}
f(\mathcal{A}^*)\!-\!f(\mathcal{A}_i)
&= f(\mathcal{A}^*)-f(\mathcal{A}_{i-1})-f_{\mathcal{A}_{i-1}}(a_i)\\
&\le \!f(\mathcal{A}^*)-f(\mathcal{A}_{i-1})
-\frac{1}{K}\bigl(f(\mathcal{A}^*\!)-f(\mathcal{A}_{i-1}\!)\bigr)\\
&=\left(1-\frac{1}{K}\right)\bigl(f(\mathcal{A}^*)-f(\mathcal{A}_{i-1})\bigr)\\
&\le \left(1-\frac{1}{K}\right)^i f(\mathcal{A}^*),  \text{    (Solving recursively)}
\end{align*}
which proves~\eqref{eq:induction}.
Using the claim for $i=K$, we get
\begin{align}
f(\mathcal{A}^*)-f(\mathcal{A}_K)
&\le
\left(1-\frac{1}{K}\right)^K f(\mathcal{A}^*) \\ \nonumber
&\le
e^{-1}f(\mathcal{A}^*), \text{       (as $K \rightarrow \infty$)} 
\end{align}
This implies $$f(\mathcal{A}_K) \geq (1 - 1/e) f(\mathcal{A}^*)$$
\qed
\section{Proof of Theorem~\ref{thm:complexity}}\label{Apdx:complexity}
At each greedy iteration, the algorithm evaluates the marginal gain of adding each candidate base station $a \in \mathcal{B} \setminus \mathcal{A}$. Computing the SINR-based coverage gain requires checking all sensors in $\mathcal{S}$.
Thus, each evaluation requires $\mathcal{O}(S)$ operations.
Since there are $B$ candidates and $K$ greedy iterations, the total computational cost is:

\[
K \times B \times S = \mathcal{O}(KBS)
\]

The space complexity is dominated by storage of sensor and base station coordinates, yielding $\mathcal{O}(B + S)$.
\qed
\bibliographystyle{ieeetr}
\bibliography{references.bib}
\end{document}